%% file: main.tex
\documentclass[11pt]{article}
\input{pretex}

\usepackage{amsmath,amsfonts,amssymb,array,graphicx,mathtools,multirow,bm,times,tcolorbox,relsize,booktabs}
\usepackage[utf8]{inputenc}
\usepackage[T1]{fontenc}
\usepackage{qcircuit}
\usepackage{authblk} 
\usepackage[paperwidth=199.8mm,paperheight=297mm,centering,hmargin=20mm,vmargin=2cm]{geometry}
\usepackage[ruled, vlined, linesnumbered]{algorithm2e}
\usepackage{pifont} 

\definecolor{colortwo}{rgb}{0.4,0.77,0.17}
\definecolor{colorthree}{rgb}{0.01,0.51,0.93}

\pgfplotsset{compat=1.18}

\newcommand{\op}[1]{\operatorname{ #1 }}

\nc{\cmark}{\ding{51}}
\nc{\xmark}{\ding{55}}

\nc{\HPTTP}{{{{HPTP-PPT}}}}
\nc{\EPPT}{{E_{\operatorname{PPT}}}}
\nc{\EPPTone}{{E_{\operatorname{PPT}}^{(1)}}}
\nc{\EK}{{E_{\kappa}}}
\nc{\EN}{E_N}
\nc{\exact}{{\operatorname{exact}}}
\nc{\PPTq}{{\operatorname{PPTq}}}
\newcommand{\trace}[2][]{\tr_{#1} \left[ #2 \right]}
\nc{\NDG}{{\operatorname{NDG}}}
\allowdisplaybreaks

\begin{document}
\title{Reversible Entanglement Beyond Quantum Operations}
\author{Xin Wang\thanks{felixxinwang@hkust-gz.edu.cn} }

\author{Yu-Ao Chen\thanks{yuaochen@hkust-gz.edu.cn} }

\author{Lei Zhang\thanks{lzhang897@connect.hkust-gz.edu.cn} }

\author{Chenghong Zhu\thanks{czhu854@connect.hkust-gz.edu.cn} }
\affil{\small Thrust of Artificial Intelligence, Information Hub,\par The Hong Kong University of Science and Technology (Guangzhou), Guangdong 511453, China}


\date{\today}
\maketitle

\begin{abstract}
We introduce a reversible theory of exact entanglement manipulation by establishing a necessary and sufficient condition for state transfer under trace-preserving transformations that completely preserve the positivity of partial transpose (PPT). Under these free transformations, we show that logarithmic negativity emerges as the pivotal entanglement measure for determining entangled states' transformations, analogous to the role of entropy in the second law of thermodynamics. Previous results have proven that entanglement is irreversible under quantum operations that completely preserve PPT and leave open the question of reversibility for quantum operations that do not generate entanglement asymptotically. However, we find that going beyond the complete positivity constraint imposed by standard quantum mechanics enables a reversible theory of exact entanglement manipulation, which may suggest a potential incompatibility between the reversibility of entanglement and the fundamental principles of quantum mechanics.
\end{abstract}

\tableofcontents


\newpage
\section{Introduction}
Reversibility is a fundamental concept in many areas of physics, including thermodynamics and quantum mechanics. The second law of thermodynamics governs the direction of heat transfer and the efficiency of energy conversion. With the existence of heat reservoirs, the second law allows for a reversible exchange of work and heat, as exemplified in the Carnot cycle~\cite{Carnot1979}. Based on axiomatic approaches and idealized conditions, it has been shown that entropy is the unique function that determines all transformations between comparable equilibrium states~\cite{Lieb1999,Giles2016}.

Within the realm of quantum information science, reversibility is crucial because it allows for the efficient manipulation of quantum resources. If a process is reversible, it implies that no quantum resources are irretrievably lost during transformations. Through the development of quantum information processing, quantum entanglement has been recognized as an essential resource for various applications, including quantum communication~\cite{Bennett1999}, quantum computation~\cite{Brus2011,Jozsa2003a}, quantum sensing~\cite{Degen2017}, and cryptography~\cite{Ekert1991}. Understanding the reversibility of entanglement is thus pivotal for quantum information and fuels the debate surrounding the axiomatization of entanglement theory. This discourse is largely driven by the parallels drawn with thermodynamics, which have fostered the potential proposition of a single entanglement measure, analogously to entropy, that could potentially govern all entanglement transformations. Such progress in understanding entanglement reversibility would not only mirror thermodynamic properties but also contribute significantly to the axiomatization of entanglement theory.

Reversibility of entanglement pertains to the process of asymptotic entanglement manipulation. For pure quantum states, this process is reversible, indicating that entanglement can be manipulated and then restored to its original state~\cite{Bennett1996b} through local operations and classical communication (LOCC). However, this asymptotic entanglement reversibility does not apply to mixed states~\cite{Vidal2001,Vidal2002b,Vollbrecht2004,Cornelio2011,Yang2005}, meaning that once entanglement is manipulated, it cannot be perfectly restored to its initial state under LOCC. The irreversibility of quantum entanglement under LOCC presents a stark contrast to the principles of thermodynamics~\cite{Horodecki2002}, where certain processes are inherently reversible. This irreversibility also underscores the impossibility of developing a single measure~\cite{Horodecki2003b} capable of governing all entanglement transformations under LOCC, suggesting that a deeper understanding of entanglement manipulation is essential for possible reversibility.

In the pursuit of a reversible entanglement theory, broader classes of operations beyond LOCC can be considered to potentially reduce the gap between entanglement cost and distillable entanglement. However, this approach has yet to yield success as Wang and Duan~\cite{Wang2016d} prove that entanglement is irreversible under quantum operations that completely maintain the positivity of partial transpose (PPT), a meaningful set of quantum operations that includes all LOCC operations. This result~\cite{Wang2016d} also implies the irreversibility in the resource theory of entanglement with non-positive partial transpose (NPT). Furthermore, Lami and Regula~\cite{Lami2023a} show that entanglement theory is irreversible under all non-entangling transformations, which are positive maps that do not produce entanglement. Notably, even for asymptotically entanglement non-generating operations~\cite{Brandao2010}, the reversibility of entanglement is not known~\cite{Berta2023,Fang2021}.
The question of whether a reversible entanglement theory could exist remains a vital open problem in quantum information theory~\cite{M.Plenio}.

In this paper, we introduce a reversible theory of exact entanglement manipulation, showing a possible counterpart of entanglement manipulation to the second law of thermodynamics.
This reversible theory operates under transformations that completely preserve the positivity of partial transpose, which are called PPT quasi-operations (PPTq operations) throughout the paper. 
Our key result is that logarithmic negativity fully determines entangled transformations under PPTq operations, i.e.,
\begin{align}
     \rho \xrightarrow{\PPTq}\sigma \iff E_N(\rho) \ge E_N(\sigma),
\end{align}
which means logarithmic negativity plays an analogous role of entropy in the second laws of thermodynamics.
Based on this necessary and sufficient condition of state transformation (cf. Theorem~\ref{thm:ent transfer}), we prove that the logarithmic negativity determines exact distillable entanglement and exact entanglement cost.
We further show the reversibility of exact entanglement manipulation under PPTq operations (cf. Theorem~\ref{thm:reversible}), showing that
\begin{align}
R_{\PPTq}(\rho_{AB}\to\sigma_{A'B'})=\frac{E_N(\rho_{AB})}{E_{N}(\sigma_{A'B'})},
\end{align}
where $R_{\PPTq}(\rho_{AB}\to\sigma_{A'B'})$ is the asymptotic ratio for exact state transformation under PPTq operations.
We further establish an inequality chain of the entanglement manipulation rates for PPTq operations~(cf.~Theorem~\ref{thm:chain inequality}), presenting a distinction between this reversible theory of entanglement beyond quantum operations and standard entanglement theory.

While our research establishes a reversible theory of exact entanglement manipulation, it's crucial to note that the allowed transformations in this theory are beyond the boundaries of quantum mechanics. But for operations in quantum mechanics, the reversibility under operations that do not generate entanglement asymptotically remains an unresolved enigma in the field~\cite{Berta2023,Fang2021}.
The above reversibility of exact entanglement manipulation under PPT quasi-operations may suggest that the coexistence of entanglement reversibility and quantum mechanics might be mutually incompatible.

\begin{figure}
    \centering
    \includegraphics[width=0.6\linewidth]{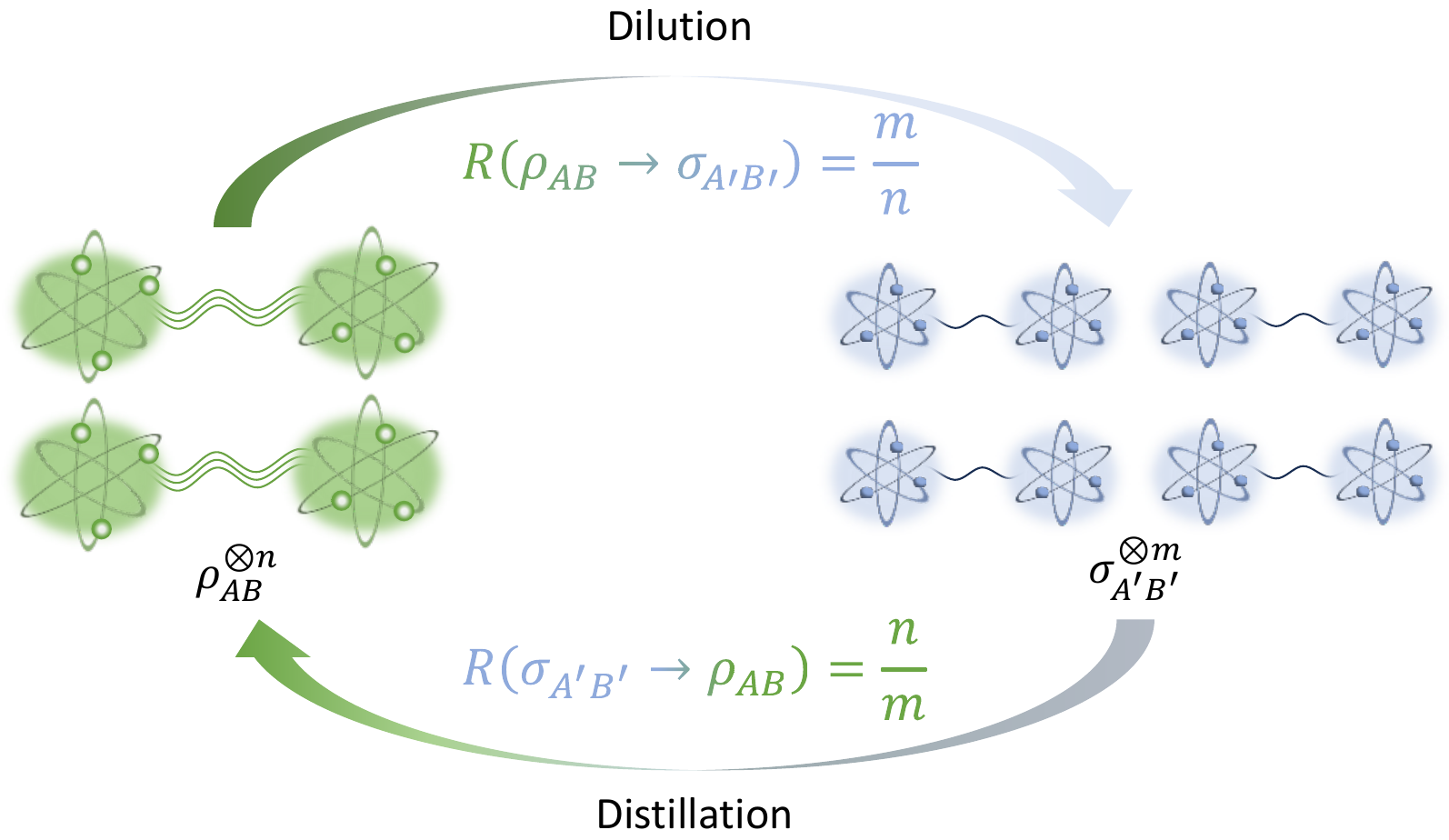}
    \caption{\textbf{Exactly reversible interconvention between quantum states $\rho_{AB}$ and $\sigma_{A'B'}$}. Consider $n$-copies of $\rho_{AB}^{\ox n}$ and $m$-copies of $\sigma_{A'B'}^{\ox m}$, the forward direction indicates the exact conversion from $\rho_{AB}^{\ox n}$ to $\sigma_{A'B'}^{\ox m}$ with rate $R_{\Omega}(\rho_{AB}^{\ox n} \to \sigma_{A'B'}^{\ox m}) = m/n$, while the backward direction shows the exact conversion from $\sigma_{A'B'}^{\ox m}$ to $\rho_{AB}^{\ox n}$ with rate $R_{\Omega}(\sigma_{A'B'}^{\ox m} \to \rho_{AB}^{\ox n}) = n/m$.
    Two states can be exactly interconverted reversibly if $R_{\Omega}(\rho_{AB}^{\ox n}\to \sigma_{A'B'}^{\ox m}) R_{\Omega}(\sigma_{A'B'}^{\ox m}\to \rho_{AB}^{\ox n}) = 1$.}
    \label{fig:enter-label}
\end{figure}
\section{Preliminaries and Entanglement Manipulation}
\subsection{Preliminaries}

\paragraph{Notations.}
We will use the symbols $A$ and $B$ to denote the finite-dimensional Hilbert spaces $\cH_A$ and $\cH_B$ associated with Alice and Bob systems, respectively.  We denote the dimension of $\cH_A$ and $\cH_B$ as $d_A$ and $d_B$.  A quantum state on system $A$ is a positive semidefinite operator $\rho_A$ with trace one.  The trace norm of $\rho$ is denoted as $\|\rho\|_1 = \tr(\sqrt{\rho^{\dagger} \rho})$. 
Let $\{\ket{i}\}_{i = 1}^{d}$ be a standard computational basis, then a standard maximally entangled state of Schmidt rank $d$ is $\Phi^d = 1/d \sum_{i,j = 1}^{d}\ketbra{ii}{jj}$. A Bell state $\Phi^2$ is alternatively called an ebit. A bipartite quantum state $\rho_{AB}$ is said to be a PPT state, if it admits positive partial transpose, i.e., $\rho_{AB}^{T_B} \geq 0$, where $T_B$ denotes the partial transpose on the system B.
Throughout the paper, we take the logarithm to be base two unless stated otherwise.

\paragraph{Properties of linear maps.} Let $\cN$ be a linear map. $\cN$ is Hermitian-preserving (HP) if it maps any Hermitian operator to another Hermitian operator; $\cN$ is positive if it maps any positive semi-definite operator to another positive semi-definite operator, and is called completely positive (CP) if this positivity is preserved on any extended reference system. $\cN$ is also trace-preserving (TP) if it preserves the traces of input operators. 

\paragraph{Properties of bipartite linear maps.} Let $\cN$ be a bipartite linear map. $\cN$ is a local operation and classical communication (LOCC) if it is a composition of a (finite) sequence of quantum instruments. $\cN$ is separability-preserving if it preserves the set of separable states. 
$\cN$ is positivity-of-the-Partial-Transpose-preserving (PPT-preserving) if it preserves the set of PPT states~\cite{Horodecki1996,Peres1996}, and is considered completely PPT-preserving (PPT) if this property holds on any extended reference system. 
We note that free operations beyond LOCC are of importance to advance our understanding of quantum entanglement (see, e.g., \cite{Rains1999,Rains2001,Eggeling2001,Brandao2010,Chitambar2017,Regula2019}).

\begin{shaded}
    \paragraph{Quasi-states.} A quasi-state is a mathematical entity that encapsulates the intricacies of a probabilistic quantum system. Despite the non-physicality in quantum mechanics, the class of quasi-states is the largest possible set that can be represented by physical quantum states under a statistical meaning. Formally, the definition of this class is given as
\begin{equation}\label{eq: def quasi state}
    \widetilde{\mathcal{S}} = \left\{ \sum\nolimits_{j} c_j \rho_j: \rho_j \in \mathcal{S}, c_j \in \RR \textrm{ s.t. } \sum\nolimits_{j} c_j = 1 \right\}
,\end{equation}
    where $\mathcal{S}$ is the set of quantum states. Note that $\widetilde{\mathcal{S}}$ is the set of all Hermitian matrices with trace one.
\end{shaded}

\subsection{Entanglement manipulation}
\paragraph{Quasi-state transformations.} 
Throughout this paper, we allow both state and quasi-state transformations, viewing them as operator transformations under linear maps. Specifically, we call a bipartite quasi-state $\rho$ can be transformed into another bipartite quasi-state $\sigma$ under $\O$ operations, provided there exists a bipartite linear map $\cN\in\O$ such that $\cN(\rho) = \sigma$. This broader framework of quasi-state transformation enables us to define distillable entanglement and entanglement cost also for bipartite quasi-states, similar to how we would for standard quantum states.  This approach offers a fresh perspective on understanding the limits of manipulating quantum entanglement.

\paragraph{Entanglement distillation.}
The maximally entangled state plays a role as the currency in quantum information since it has become a key ingredient in many quantum information processing tasks (e.g., teleportation~\cite{Bennett1993}, superdense coding~\cite{Bennett1992}, and quantum cryptography~\cite{Ekert1991}).
It is important to understand how many maximally entangled states we can obtain from a source of less entangled states using free operations. Imagine that Alice and Bob share a large supply of identically prepared states, and they want to convert these states to high-fidelity Bell pairs.
Let $\Omega$ represent a set of free operations or allowed transformations.
The one-shot zero-error, or exact distillable entanglement, under $\O$ operations of quantum state or quasi-state $\rho_{AB}$ is defined as
$E^{(1)}_{0,D,\O}(\rho_{AB})= \sup_{\Lambda\in \Omega}\left\{\log d:   \Phi^{d}_{\hat{A}\hat{B}}=\Lambda_{AB\to\hat{A}\hat{B}} \left(\rho_{AB}\right)\right\}$,
where  $\Phi^{d}_{\hat{A}\hat{B}}=[1/d]\sum_{i,j=1}^{d}\ketbra{ii}{jj}_{\hat{A}\hat{B}}$ represents the standard maximally entangled state of Schmidt rank~$d$.
The zero-error, or exact distillable entanglement, of a bipartite state or quasi-state state $\rho_{AB}$, under $\Omega$ operations is defined as
\begin{align}
E_{D,\O}^{\exact}(\rho_{AB})= \lim_{n \to \infty} \frac{1}{n}E^{(1)}_{0,D,\O}\left(\rho_{AB}^{\ox n}\right).
\end{align}
For entanglement distillation with  asymptotically vanishing error, the rate is quantified via distillable entanglement. The distillable entanglement of a bipartite state or quasi-state $\rho_{AB}$, under the $\Omega$ operations, is defined as
\begin{align}
E_{D,\O}\left(\rho_{AB}\right)=\sup\left\{r: \lim_{n \to \infty} \left[\inf_{\Lambda\in\O}  \norm{ \Lambda\left(\rho_{AB}^{\ox n}\right)- \Phi^{2rn}_{\hat{A}\hat{B}} }_1 \right]=0\right\}.
\end{align}
As distillable entanglement quantifies the fundamental limit of entanglement distillation and related task of quantum communication, substantial efforts have been made to obtain its accurate estimation and fundamental properties~\cite{Rains2001,Eggeling2001,Horodecki2000a,Christandl2004,Wang2016,Wang2016c,Leditzky2017,Wang2019b,Zhu2023}. As exact entanglement distillation is more restricted, it holds that
\begin{align}\label{eq:measure ED inequality}
E_{D,\Omega}^{\exact} (\rho) \le E_{D,\Omega}(\rho).
\end{align}

\paragraph{Entanglement dilution.}
The reverse task of entanglement distillation is called entanglement dilution.  At this time, Alice and Bob share a large supply of Bell pairs and they are to convert $rn$ Bell pairs to $n$ high fidelity copies of the desired state $\rho^{\ox n}$ using suitable free operations. Let $\Omega$ represent a set of free operations, which for example can be \text{LOCC} or \text{PPT}. 
The one-shot zero-error, or exact entanglement cost, of a bipartite state or quasi-state $\rho_{AB}$ under the $\Omega$ operations is defined as
$E^{(1)}_{0,C,\O}\left(\rho_{AB}\right)= \inf_{\Lambda\in \Omega}\left\{\log d:   \rho_{AB}=\Lambda_{\hat{A}\hat{B}\to AB} \left(\Phi^{d}_{\hat{A}\hat{B}}\right)\right\}$.
The zero-error, or exact entanglement cost, of  $\rho_{AB}$ under the $\Omega$ operations is defined as
\begin{align}
E^{\exact}_{C,\O}\left(\rho_{AB}\right)= \lim_{n \to \infty} \frac{1}{n} E^{(1)}_{0,C,\O}\left(\rho_{AB}^{\ox n}\right).
\end{align}
Previous works~\cite{Terhal2000b,Audenaert2003,Matthews2008,Wang2020c} have shown progresses towards understanding the exact entanglement cost of quantum states. 

For the case of asymptotically vanishing error of entanglement dilution, the rate is quantified via entanglement cost. The concise definition of entanglement cost using $\O$ operations is given as follows:
\begin{align}
E_{C,\Omega} \left(\rho_{AB}\right) = \inf\left\{r: \lim_{n \to \infty} \inf_{\Lambda\in \Omega}  \|\rho_{AB}^{\ox n}-\Lambda \left(\Phi^{2rn}_{\hat{A}\hat{B}}\right)\|_1=0\right\},
\end{align}
When LOCC is free, entanglement cost is given by the regularization of the entanglement of formation~\cite{Hayden2001}, which is shown to be non-additive~\cite{Hastings2009}.
Further efforts have been made to improve understanding of the entanglement cost in specific and general quantum states.~\cite{Christandl2012,Wang2016d,Wilde2018,Wang2023b}. As exact entanglement dilution is more restricted, it holds that
\begin{align}\label{eq:measure EC inequality}
E_{C,\Omega}(\rho) \le E_{C,\Omega}^{\exact}(\rho).
\end{align}

\paragraph{Exact entanglement transformations.}
The ratio of state conversion plays an integral role in the manipulation of quantum resources such as entanglement.
In this paper, we mainly focus on the asymptotic exact state conversion.
It describes the process of converting one state into another exactly as the number of copies approaches infinity under certain free operations.
Let $\Omega$ represent a set of free operations.
The asymptotic conversion ratio of exact entanglement transformation from a state or quasi-state $\rho_{AB}$ to another state or quasi-state  $\sigma_{\hat{A}\hat{B}}$ is defined as
\begin{align}
    R_\Omega \left(\rho_{AB}\to\sigma_{A'B'} \right) = \sup \left\{\frac{m}{n}: \exists n_0: \forall n \geq n_0, \exists \Lambda_n \in  \Omega: \Lambda_n\left(\rho_{A B}^{\otimes n}\right)=\sigma_{A'B'}^{\ox m}\right\}
.\end{align}
Two states can be exactly interconverted reversibly if 
$R_\Omega\left(\rho_{AB}\to\sigma_{A'B'}\right) \times R_\Omega\left(\sigma_{A'B'}\to\rho_{AB}\right) = 1$.

\begin{figure}[t]
    \centering
    \includegraphics[width=0.6\linewidth]{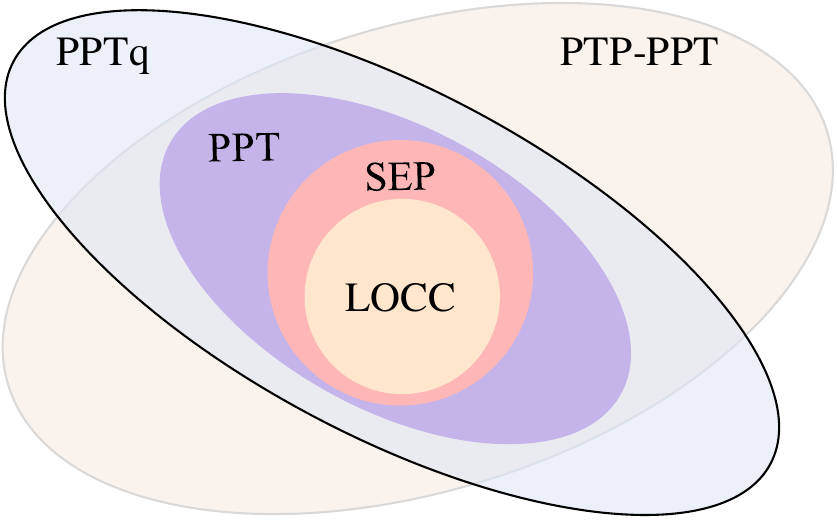}
    \begin{tabular}{cc}
    \toprule
    \midrule
    Class of Operations & Abbreviation \\
    \midrule
    Local Operations and Classical Communication & LOCC \\
    Separable Operations & SEP \\
    Completely PPT-preserving  Operations & PPT \\
    Positive and Trace-preserving and PPT-preserving Maps  & PTP-PPT  \\
    Hermitian-preserving and Trace-preserving PPT Maps & PPTq \\
    \addlinespace
    \midrule
    \bottomrule
    \end{tabular}
    \caption{Schematic hierarchy of operations. The pictured inclusions among $\LOCC$, $\SEP$, $\PPT$, $\textbf{\rm PTP-PPT}$, $\textbf{\rm PPTq}$. The main focus of this paper is entanglement manipulation under PPTq operations.}
    \label{fig:hierarchy}
\end{figure}

\section{Reversibility of Exact Entanglement Manipulation}
\subsection{Necessary and sufficient condition for entanglement transformations}

In this section, our focus is centered in exploring the possibility of reversible entanglement manipulation. To thoroughly characterize the fundamental limits of entanglement transformation, we briefly venture beyond the confines of the quantum realm in the context of this section, focusing on entanglement transformations beyond quantum operations. 
Following the idea of axiomatic approaches, we consider transformations that demand the weakest possible requirements for entanglement manipulation, relaxing the manipulating objects from quantum states to quasi-states~(c.f.~Eq.~\eqref{eq: def quasi state}).

We specifically introduce the PPT quasi-operations (PPTq), which are Hermitian-preserving and trace-preserving maps that completely preserve the positive of partial transpose. The underlying intuition is that PPT states are all bound entanglement~\cite{Horodecki1998}, which are useless in entanglement distillation. Furthermore, as depicted in Figure~\ref{fig:hierarchy}, PPTq operations encompass all SEP, LOCC, and PPT operations, whereas some PTP-PPT operations do not fall within the scope of $\PPTq$.

\begin{shaded}
\begin{definition}
    An HPTP bipartite map $\cN_{AB\to A'B'}$ is called a PPT quasi-operation (PPTq operation) if $T_{B'} \circ \cN_{AB\to A'B'} \circ T_B$ is completely positive.
\end{definition}
\end{shaded}

We are going to firstly establish the necessary and sufficient condition for perfect transformations between quasi-states. It turns out surprising that 
the logarithmic negativity~\cite{Vidal2002,Plenio2005b} is the key to fully characterize the transformations under PPTq operations. We also note that the following theorem directly applies to the more restricted case for transformations between quantum states.

\begin{tcolorbox}[width = 1.0\linewidth]
\begin{theorem}\label{thm:ent transfer}
    For two bipartite states or quasi-states $\rho$ and $\sigma$, there exists $\cN_{AB \to A'B'} \in \operatorname{PPTq}$ such that $\cN(\rho) = \sigma$ if and only if
\begin{equation}
    \EN(\rho) \geq \EN(\sigma),
\end{equation}
where $\EN(\rho) = \log\|\rho^{T_B}\|_1$ is the logarithmic negativity.
\end{theorem}
\end{tcolorbox}
\begin{proof}
($\implies$) Construct a linear map $\cM = T_{B'} \circ \cN \circ T_B$. Since $\cN_{AB \to A'B'} \in \operatorname{PPTq}$, we know that $\cM$ is completely positive and trace-preserving. By construction, it holds that $\cM\left(\rho^{T_B}\right) = \sigma^{T_{B'}}$.
Since any CPTP operation does not increase the trace norm, we immediately have that $\norm{\rho^{T_B}}_1 \geq \norm{\sigma^{T_{B'}}}_1$
and hence $\EN(\rho) \geq \EN(\sigma)$ by the monotonicity of the $\log$ function.

($\impliedby$) The key idea to show this direction is to construct a CPTP map $\cM$ such that $\cM\left(\rho^{T_B}\right) = \sigma^{T_{B'}}$ based on the assumption that $\|\rho^{T_B}\|_1 \ge \|\sigma^{T_{B'}}\|_1$, which could guarantee that
$\cN = T_{B'} \circ \cM \circ T_B$ is a HPTP and PPT map that can successfully transform $\rho$ to $\sigma$. 
Consider the spectral decomposition $\rho^{T_B} = \sum_j r_j \ketbra{j}{j} = R_+ - R_-$ with
$R_+ = \sum_{j: r_j \geq 0} r_j \ketbra{j}{j}$ and $R_- = -\sum_{j: r_j < 0} r_j \ketbra{j}{j}$.
Here we denote the projections for positive and negative parts as 
\begin{equation}
P_+ = \sum_{j: r_j \geq 0} \ketbra{j}{j}, \quad P_- = \sum_{j: r_j < 0} \ketbra{j}{j}, \quad P_+ + P_- = I.
\end{equation}
Without loss of generality, we assume that $\EN(\sigma) > 0$. Let us then denote $\sigma^{T_{B'}} = \sum_n s_n \ketbra{n}{n}$ and choose some $k \in \NN$ such that $s_k \geq 0$. We further could write $\sigma^{T_{B'}} = \widetilde{S}_+ - S_- + s_k \ketbra{k}{k}$ with
\begin{equation}    
    \widetilde{S}_+ = \sum_{n: s_n \geq 0,\, n \neq k} s_n \ketbra{n}{n}, \quad
    S_- = -\sum_{n: s_n < 0} s_n \ketbra{n}{n}.
\end{equation}

We construct the following CPTP map:
\begin{align}\label{eq: the channel M}
    \cM(\o) = \frac{\tr(P_+\o)}{\tr R_+}  \widetilde{S}_+ + \frac{\tr(P_-\o)}{\tr R_-} S_- + \lambda\ketbra{k}{k},
\end{align}
with 
$\lambda = \tr \o - {\tr(P_+\o)\tr \widetilde{S}_+}/{\tr R_+} - {\tr(P_-\o)\tr S_-}/{\tr R_-}$.
This construction guarantees the condition of TP as the trace of the R.H.S. of Eq.~\eqref{eq: the channel M} is equal to $\tr \o$. To see that this map is also CP, we only need to show that $\lambda\ge 0$. Note that the prerequisite $\EN(\rho) \geq \EN(\sigma)$ implies that
$\frac{\tr \widetilde{S}_+}{\tr R_+} \le 1$ and $\frac{\tr S_-}{\tr R_-} \leq 1$, thus we have that
\begin{align}
\lambda \ge \tr \o - \tr(P_+\o) - \tr(P_-\o) = 0.
\end{align}

Finally, we apply the quantum channel $\cM$ in Eq.~\eqref{eq: the channel M} to $\rho^{T_B}$ and obtain 
\begin{align}
    \cM(\rho^{T_B})  =  \cM( R_+ - R_-) 
     =  \widetilde{S}_+ - S_- + s_k\ketbra{k}{k} = \sigma^{T_B}.
\end{align}
\end{proof}

This result implies that logarithmic negativity emerges as the pivotal entanglement measure for determining the transformations between entangled states, which is analogous to the role of entropy in the second law of thermodynamics.

\subsection{Reversibility of exact entanglement manipulation under PPTq operations}

Entanglement distillation and entanglement dilution are two vital operational tasks in entanglement manipulation. Entanglement distillation involves the transformation of a large number of identically and independently distributed (i.i.d.) copies of a certain state into as many ebits as possible. Conversely, entanglement dilution is concerned with the reverse process, turning ebits into as many copies of the original state. These procedures are typically carried out by means of LOCC. Here, we are going to solve the key rates for both exact distillable entanglement and exact entanglement cost under PPTq operations, which are shown to be equal to the logarithmic negativity of the state.

\begin{tcolorbox}[width = 1.0\linewidth]
\begin{theorem}\label{thm:ent measure}
    For any bipartite state or quasi-state $\rho_{AB}$, it holds that
\begin{align}
E_{D,\operatorname{PPTq}}^{\operatorname{exact}}(\rho_{AB})  =E_{C,\operatorname{PPTq}}^{\operatorname{exact}}(\rho_{AB}) = E_N({\rho_{AB}}).
\end{align}
\end{theorem}
\end{tcolorbox}
\begin{proof}
The proof of this theorem is divided into two parts. We are going to firstly show $E_N(\rho_{AB}) \le E_{D,\operatorname{PPTq}}^{\operatorname{exact}}(\rho_{AB})$ and then show $E_{C,\operatorname{PPTq}}^{\operatorname{exact}}(\rho_{AB})\le E_N(\rho_{AB})$.

Let us consider $n$-copy of the state $\rho_{AB}$ and try to transform it to as many ebits as possible.
By Theorem~\ref{thm:ent transfer}, we know that there exists one $\Lambda\in\PPTq$ and maximally entangled state $\Phi_{A'B'}^{d}$ with $d = \lfloor 2^{n{E_N(\rho_{AB})}} \rfloor$ such that
\begin{align}
    \Lambda(\rho_{AB}^{\otimes n}) = \Phi_{A'B'}^{d},
\end{align}
since $E_N(\rho_{AB}^{\otimes n}) = nE_N(\rho_{AB}) \ge  \log \lfloor 2^{n{E_N(\rho_{AB})}} \rfloor = \log d = E_N(\Phi_{A'B'}^{d})$.
Thus, by the definition of exact entanglement distillation, we have
\begin{align}
E^{(1)}_{0,D,\PPTq}(\rho_{AB}^{\otimes n}) \ge \log \lfloor 2^{n{E_N(\rho_{AB})}} \rfloor,
\end{align}
which leads to
\begin{align}\label{eq:ED EN}
    E_{D,\PPTq}^{\exact}(\rho_{AB}) &= \lim_{n \to \infty} \frac{1}{n}E^{(1)}_{0,D,\PPTq}(\rho_{AB}^{\otimes n})\ge \lim_{n \to \infty} \frac{1}{n}\log \lfloor 2^{n{E_N(\rho_{AB})}} \rfloor 
     = E_N(\rho_{AB}).
\end{align}
 
For the reverse task of entanglement dilution, 
let us consider $n$-copy of the state $\rho_{AB}$ and try to transform at least ebits as possible to prepare $\rho_{AB}^{\ox n}$.
By Theorem~\ref{thm:ent transfer}, we know that there exists one $\Lambda\in\PPTq$ and maximally entangled state $\Phi_{A'B'}^{d}$ with $d = \lceil 2^{n{E_N(\rho_{AB})}} \rceil$ such that
\begin{align}
    \Lambda(\Phi_{A'B'}^{d}) = \rho_{AB}^{\otimes n},
\end{align}
since $E_N(\rho_{AB}^{\otimes n}) = nE_N(\rho_{AB}) \le  \log \lceil 2^{n{E_N(\rho_{AB})}} \rceil = \log d = E_N(\Phi_{A'B'}^{d})$.
Therefore, by the definition of exact entanglement cost, we have
\begin{align}
E^{(1)}_{0,C,\PPTq}(\rho_{AB}^{\otimes n}) \le \log \lceil 2^{n{E_N(\rho_{AB})}} \rceil,
\end{align}
which leads to
\begin{align}\label{eq:EC EN}
    E_{C,\PPTq}^{\exact}(\rho_{AB}) &= \lim_{n \to \infty} \frac{1}{n}E^{(1)}_{0,C,\PPTq}(\rho_{AB}^{\otimes n})\le \lim_{n \to \infty} \frac{1}{n}\log \lceil 2^{n{E_N(\rho_{AB})}} \rceil 
     = E_N(\rho_{AB}).
\end{align}

Note that exact transformations between $\rho$ and $\Phi^2$ guarantees the inequality
\begin{align}\label{eq:ED EC exact}
E_{D,\PPTq}^{\operatorname{exact}}(\rho) \le E_{C,\PPTq}^{\operatorname{exact}}(\rho). 
\end{align}
Thus, combining Eq.~\eqref{eq:ED EC exact}, Eq.~\eqref{eq:ED EN}  and Eq.~\eqref{eq:EC EN}, we arrive at
\begin{align}
E_{D,\operatorname{PPTq}}^{\operatorname{exact}}(\rho_{AB})  =E_{C,\operatorname{PPTq}}^{\operatorname{exact}}(\rho_{AB}) = E_N({\rho_{AB}}).
\end{align}
\end{proof}

The above result has already implies a collapse of two important entanglement measures. As the logarithmic negativity quantifies both exact entanglement cost and exact distillable entanglement, we want to note that it also explains why logarithmic negativity has been fruitfully used in the theory of quantum entanglement. 
For example, previous works have shown that the logarithmic negativity is an upper
bound to distillable entanglement~\cite{Vidal2002} and possesses an operational interpretation as the exact entanglement cost under PPT operations for certain classes of quantum states~\cite{Audenaert2003,Ishizaka2004a}.

Furthermore, using results on exact entanglement cost and exact distillable entanglement, we are now able to prove the reversibility of asymptotic exact entanglement manipulation under PPTq operations.
\begin{tcolorbox}[width = 1.0\linewidth]
\begin{theorem}\label{thm:reversible}
    For any two bipartite states $\rho_{AB}$ and $\sigma_{A'B'}$, the asymptotic exact entanglement transformation rate is given by
\begin{align}
R_{\operatorname{PPTq}}(\rho_{AB}\to\sigma_{A'B'})=\frac{E_N(\rho_{AB})}{E_{N}(\sigma_{A'B'})},
\end{align}
which implies the reversibility of asymptotic exact entanglement manipulation under $\PPTq$ operations, i.e.,
\begin{align}
    R_{\operatorname{PPTq}}(\rho_{AB}\to\sigma_{A'B'}) \times R_{\operatorname{PPTq}}(\sigma_{A'B'}\to \rho_{AB}) = 1.
\end{align}
\end{theorem}
\end{tcolorbox}
\begin{proof}
The key idea of this proof is to use the maximally entangled state for intermediate exchange between states, that is,
\begin{align}
R_\PPTq(\rho_{AB}\to\Phi_{\hat{A}\hat{B}}^2)=E_{D,\PPTq}^{\exact}(\rho_{AB}) = E_N(\rho_{AB}),
\end{align}
\begin{align}
R_\PPTq(\Phi_{\hat{A}\hat{B}}^2\to \sigma_{A'B'})=E_{C,\PPTq}^{\exact}(\sigma_{A'B'})^{-1} = E_N(\sigma_{A'B'})^{-1}.
\end{align}
Therefore, we could obtain the transformation ratio as
\begin{align}
    R_{\operatorname{PPTq}}(\rho_{AB}\to\sigma_{A'B'}) &= R_\PPTq(\rho_{AB}\to\Phi_{\hat{A}\hat{B}}^2)\times R_\PPTq(\Phi_{\hat{A}\hat{B}}^2\to \sigma_{AB})\\
     &= E_N(\rho_{AB})\times E_N(\sigma_{A'B'})^{-1}.
\end{align}
\end{proof}

As we have established the tight connection for exact measures $E_{C,\PPTq}^\exact(\rho)$ and $E_{D,\PPTq}^\exact(\rho)$ in Theorem~\ref{thm:ent measure}, we finally arrive at the inequality chain of the entanglement manipulation rates for $\PPTq$ operations as shown the following theorem.

\begin{tcolorbox}[width = 1.0\linewidth]
\begin{theorem} \label{thm:chain inequality}
    For any bipartite state $\rho$, it holds that
    \begin{align}
        E_N^\tau(\rho) \le E_{C,\PPTq}(\rho) \le E_{C,\PPTq}^\exact(\rho) =
        E_{N}(\rho) = E_{D,\PPTq}^\exact(\rho) = 
        E_{D,\PPTq}(\rho). 
    \end{align}
\end{theorem}
\end{tcolorbox}
\noindent\emph{Sketch of Proof.} Note that $E_{C,\PPTq}(\rho) \leq E_{C,\PPTq}^\exact(\rho)$ and $E_{D,\PPTq}^\exact(\rho) \leq E_{D,\PPTq}(\rho)$ follow by Equation~\eqref{eq:measure EC inequality} and Equation~\eqref{eq:measure ED inequality}, respectively. Since Theorem~\ref{thm:ent measure} bridges the gap between these two exact measures, it is sufficient to prove the two endpoints in the inequality chain. A full proof can be found in
Appendix~\ref{appendix:chain inequality}.

\begin{remark}
In the standard quantum resource theory, the achievable rate of entanglement dilution typically surpasses that of entanglement distillation. However, the introduction of $\PPTq$ operations reverses this kind of operational inequality. The uncanny phenomenon is attributed to the unique property of quasi-operations, which can be decomposed into positive and negative operation components. Such property allows $\PPTq$ operations to ``borrow'' additional entanglement resources from seemingly out of nowhere. Specifically, these components can generate states with extra distillable entanglement. 
Once the desired transformations are accomplished, this borrowed resource can then be effectively ``returned'' by combining the positive and negative components to form the target state. This unique feature allows for the achievable rate of entanglement distillation to surpass that of entanglement dilution, presenting a surprisingly reversal of the usual operational inequality in standard resource theory.
    
Notably, even though the proof demonstrates $E_{C,\PPTq}(\rho)\le E_{D,\PPTq}(\rho)$, the absence of asymptotic continuity in the logarithmic negativity $E_N$ suggests that $E_{C,\PPTq}(\rho)$ does not necessarily equal to $E_N(\rho)$. This observation is indicative of the fact that the resource theory in $\PPTq$ remains compatible with the existing entanglement theories under quantum operations.
\end{remark}

\section{Discussions}
Our work has demonstrated the reversibility of exact entanglement transformations under PPTq operations. This reversibility establishes a parallel between entanglement manipulation and the second law of thermodynamics, particularly when operating under idealized conditions. The logarithmic negativity
is the key entanglement measure in this reversible entanglement theory to determine the exact transformation between entangled states, serving a role analogous to entropy in the realm of thermodynamics. 
The advent of a reversible theory of exact entanglement manipulation under PPTq operations paves the way for further exploration into the smallest subset of quantum operations or maps nestled within the set of PPTq operations that could guarantee the reversibility of asymptotic entanglement manipulation.  Our work also opens a possible avenue to study quantum resources beyond free quantum operations.

\begin{table}[htbp]
\centering
\resizebox{\textwidth}{!}{
\begin{tabular}{lcccc}
\toprule
\midrule
Paradigm & Class of Operations & Free Resource & Resource & Reversible? \\
\midrule
 Thermodynamics & - & Heat & Work & \cmark \\
 Coherence~\cite{Winter2016} & Incoherent & Incoherent States & Coherent States & \xmark \\
 Coherence~\cite{Brandao2015} & Maximally Incoherent & Incoherent States & Coherent States & \cmark \\
 Entanglement~\cite{Vidal2001} & LOCC & Separable States & Entangled States & \xmark \\
 NPT Entanglement~\cite{Wang2016d} & PPT & PPT States & NPT States & \xmark \\
 Entanglement~\cite{Lami2023a} & Non-Entangling  & Separable States & Entangled States & \xmark \\
 NPT Entanglement (\textbf{This Work}) & PPTq & Quasi-states w. zero $E_N$ & Quasi-states w. positive $E_N$ & \cmark \\
\addlinespace
\midrule
\bottomrule
\end{tabular}}
\caption{Mainstream resource theories.  This work presents a reversible theory of exact entanglement manipulation under PPTq operations, where asymptotic exact entanglement transformation is reversible.}
\label{tab:paradigm comp} 
\end{table}

While our work delineates a reversible theory of exact entanglement manipulation, it is imperative to remark that these transformations ostensibly transcend the standard quantum mechanics as the allowed operations extend beyond quantum operations. Meanwhile, within the domain of quantum mechanics, the reversibility of entanglement under quantum operations that do not asymptotically generate entanglement (ANE) remains an open question~\cite{Berta2023,Fang2021} due to a gap in the proof of the generalised quantum
Stein’s lemma. The reversibility of exact entanglement manipulation along with the uncertainty of reversibility under ANE quantum operations may suggest a potential incompatibility between the foundational principles of quantum mechanics and the reversibility of quantum entanglement.  
In addition, as a recent work~\cite{Regula2023} shows that the reversibility of quantum resources could happen when relaxing to probabilistic transformations, it will also be interesting to study the interplay between reversibility, success probability, and positivity of allowed transformations of quantum resources.

\section*{Acknowledgements}
We would like to thank Chengkai Zhu for helpful discussions and thank Benchi Zhao, Xuanqiang Zhao, Mingrui Jing for their useful comments. 
This work has been supported by the Start-up Fund from The Hong Kong University of Science and Technology (Guangzhou).  

\small
\bibliographystyle{alpha}
\bibliography{Ref}

\newpage
\appendix
\setcounter{subsection}{0}
\setcounter{table}{0}
\setcounter{figure}{0}

\vspace{3cm}

\begin{center}
\Large{\textbf{Appendix for Reversible Entanglement Beyond Quantum Operations} \\ \textbf{
}}
\end{center}

\renewcommand{\theequation}{S\arabic{equation}}
\renewcommand{\thesubsection}{\normalsize{Supplementary Note \arabic{subsection}}}
\renewcommand{\theproposition}{S\arabic{proposition}}
\renewcommand{\thedefinition}{S\arabic{definition}}
\renewcommand{\thefigure}{S\arabic{figure}}
\setcounter{equation}{0}
\setcounter{table}{0}
\setcounter{section}{0}
\setcounter{proposition}{0}
\setcounter{definition}{0}
\setcounter{figure}{0}


\section{Appendix: Inequality chain for entanglement measures under PPTq}\label{appendix:chain inequality}
In the appendix, we will rigorously demonstrate how the inequality chain in Theorem~\ref{thm:chain inequality} 
\begin{equation}
    E_N^\tau(\rho) \le E_{C,\PPTq}(\rho) \le E_{C,\PPTq}^\exact(\rho) =
    E_{N}(\rho) = E_{D,\PPTq}^\exact(\rho) = 
    E_{D,\PPTq}(\rho).
\end{equation}
is build. 
We start from the construction of two endpoints in the chain, i.e.\ prove $E_N^\tau(\rho) \leq E_{C,\PPTq}(\rho_{AB})$ (Proposition~\ref{prop:temp lowerbound}) and $E_{D,\PPTq}(\rho_{AB}) \leq E_N(\rho_{AB})$ (Proposition~\ref{prop:ent distill <= en}).
For the temper logarithm negativity part, we first need to extend the definition of tempered negativity in Ref.~\cite{Lami2023a} to fit in the regime of quasi-states.

\begin{definition}
    Let $\sigma_{AB}, \rho_{AB}$ be two quasi-states. The tempered negativity between $\sigma_{AB}$ and $\rho_{AB}$ is defined as
\begin{equation}
    N_{\tau} \left(\sigma_{AB} \mid \rho_{AB} \right) = \sup\{ \trace{X\sigma_{AB}}  : \norm{ X^{T_B} }_\infty \leq 1 , \| X\|_\infty = \trace{ X\rho_{AB} } \}
.\end{equation}
    Further, the tempered negativity of $\rho_{AB}$ is denoted by
\begin{equation}
    N_{\tau} \left(\rho_{AB}\right) = N_{\tau} \left(\rho_{AB} \mid \rho_{AB} \right)
.\end{equation}
\end{definition}
Here, we restate three properties of $N_\tau$, and show that the original proof in Ref.~\cite{Lami2023a} remains true after the extension. 
Then we are ready to give the proposition of inequality.

\begin{lemma} \label{lem:temp properties}
    For any quasi-state $\omega_{AB}$ and state $\rho_{AB}$,
\begin{enumerate}
    \item [$\mathbf{(a)}$] $\norm{ \sigma_{AB}^{T_B} }_1 \geq N_\tau(\sigma_{AB} \mid \rho_{AB})$,
    \item [$\mathbf{(b)}$] $\norm{\sigma_{AB} - \rho_{AB}}_1 \leq \eps \implies N_\tau\left( \sigma_{AB} \mid \rho_{AB} \right) \geq (1 -  \eps) N_\tau \left( \rho_{AB} \right)$, and
    \item [$\mathbf{(c)}$] $N_\tau\left( \rho_{AB}^{\ox n} \right) \geq N_\tau\left( \rho_{AB} \right)^n$.
\end{enumerate}
\end{lemma}
\begin{proof}
    $\mathbf{(a)}$ This property follows by the fact that
    $\norm{ \sigma_{AB} }_1 = \sup\{ \trace{ X \sigma_{AB} }  : \norm{ X^{T_B} }_\infty \leq 1 \}$.

    $\mathbf{(b)}$ Note that the H\"{o}lder's inequality holds for arbitrary complex matrices. Hence, this properties still holds by Equation (S46) in Ref.~\cite{Lami2023a}.

    $\mathbf{(c)}$ Since $\rho_{AB}$ is a quantum state, this property holds from the same reasoning in Ref.~\cite{Lami2023a}.
\end{proof}

\begin{proposition} \label{prop:temp lowerbound}
    For any bipartite state $\rho_{AB}$, it holds that
    \begin{align}
        E_{C,\PPTq}(\rho_{AB}) \geq E_N^\tau(\rho)
    .\end{align}
\end{proposition}
\begin{proof}
    This proof mainly follows from Lemma~\ref{lem:temp properties} and the idea of chain inequalities in Ref.~\cite{Lami2023a}. For $r = E_{C,\PPTq}(\rho_{AB})$, there exists a sequence of PPT quasi-operations $\{ \Lambda_{n} \}_n$ such that
\begin{equation}
    \sigma_{n, \hat{A}\hat{B}} = \Lambda_{n} \left( \Phi_{\hat{A}\hat{B}}^{2 \floor{rn}} \right),
    \quad \eps_n = \norm{ \sigma_{n, \hat{A}\hat{B}} - \rho_{AB}^{\ox n} }_1 \textrm{ and } \lim_{n \to \infty} \eps_n = 0
.\end{equation}
    Then for all $n$, Theorem~\ref{thm:ent transfer} implies $E_N \left( \Phi_{\hat{A}\hat{B}}^{2 \floor{rn}} \right) \geq E_N \left( \sigma_{n, \hat{A}\hat{B}} \right)$ and hence
\begin{equation}
\begin{aligned}
    2^{\floor{rn}}  
    &\geq \norm{ \left( \sigma_{n, \hat{A}\hat{B}} \right)^{T_B} }_1 \\
    &\overset{\mathbf{(a)}}{\geq} N_\tau \left( \sigma_{n, \hat{A}\hat{B}} \mid \rho_{AB}^{\ox n} \right)
    \overset{\mathbf{(b)}}{\geq} \left( 1 - \eps_n \right) N_\tau \left( \rho_{AB}^{\ox n} \right) 
    \overset{\mathbf{(c)}}{\geq} \left( 1 - \eps_n \right) N_\tau \left( \rho_{AB} \right)^{n}
.\end{aligned}
\end{equation}
Taking $n \to \infty$ on both sides gives $r = \lim_{n \to \infty} \floor{rn} / n \geq E_N^\tau \left( \rho_{AB} \right)$.
\end{proof}

For the entanglement distillation part, the proposition can be given as follows.

\begin{proposition} \label{prop:ent distill <= en}
    For any state or quasi state $\rho_{AB}$, it holds that
    \begin{align}
        E_{D,\PPTq}(\rho_{AB}) \leq E_N(\rho_{AB}).
    \end{align}
\end{proposition} 
\begin{proof}
    This proof mainly follows the proof from Ref.~\cite{Vidal2002}. 
    For $r = E_{D,\PPTq}(\rho_{AB})$, there exists a sequence of PPT quasi-operations $\{ \Lambda_{n} \}_n$ such that
\begin{equation}
    \sigma_{n, \hat{A}\hat{B}} = \Lambda_{n} \left( \rho_{AB}^{\ox n} \right),
    \quad \eps_n = \norm{ \sigma_{n, \hat{A}\hat{B}} - \Phi^{2\floor{rn}}_{\hat{A}\hat{B}} }_1 \textrm{ and } \lim_{n \to \infty} \eps_n = 0
.\end{equation}
    Then for all $n$, Theorem~\ref{thm:ent transfer} implies $E_N \left( \rho_{AB}^{\ox n } \right) \geq E_N \left( \sigma_{n, \hat{A}\hat{B}} \right)$ and hence
\begin{align}
     \norm{ \left(\rho_{AB}^{\ox n }\right)^{T_B} }_1 &\geq \norm{ \left(\sigma_{n, \hat{A}\hat{B}}\right)^{T_B} }_1
    = \norm{ \left(\Phi^{2\floor{rn}}_{\hat{A}\hat{B}} \right)^{T_B} + \left( \sigma_{n, \hat{A}\hat{B}} - \Phi^{2\floor{rn}}_{\hat{A}\hat{B}} \right)^{T_B} }_1 \\
    &\geq \norm{ \left(\Phi^{2\floor{rn}}_{\hat{A}\hat{B}} \right)^{T_B} }_1 - \norm{ \left( \sigma_{n, \hat{A}\hat{B}} - \Phi^{2\floor{rn}}_{\hat{A}\hat{B}} \right)^{T_B} }_1 \\
    &\geq \left( 1 - \eps_n \right)2^{\floor{rn}}
,\end{align}
    where the last inequality follows from the fact 
    $\norm{X^{T_B}}_1 \leq d \norm{X}_1$ for $X \in \op{Herm}\left( \cH_d \right)$. Taking $n \to \infty$ on both sides, the additivity of logarithm negativity gives 
\begin{equation}
    E_N \left( \rho_{AB} \right) 
    \geq \lim_{n \to \infty} \frac{1}{n} \log \left( 1 - \eps_n \right)2^{\floor{rn}} 
    = \lim_{n \to \infty} \frac{\floor{rn}}{n} + \frac{\log \left( 1 - \eps_n \right)}{n} 
    = r
.\end{equation}
\end{proof}

We are ready to construct the whole chain.

\renewcommand\theproposition{\ref{thm:chain inequality}}
\setcounter{proposition}{\arabic{proposition}-1}
\begin{theorem}
    For any bipartite state $\rho$, it holds that
    \begin{align}
        E_N^\tau(\rho) \le E_{C,\PPTq}(\rho) \le E_{C,\PPTq}^\exact(\rho) =
        E_{N}(\rho) = E_{D,\PPTq}^\exact(\rho) = 
        E_{D,\PPTq}(\rho). 
    \end{align}
\end{theorem}
\renewcommand{\theproposition}{\arabic{proposition}}
\begin{proof}
    For all bipartite state $\rho$,
\begin{alignat}{2}
    E_N^\tau(\rho) 
    &\leq E_{C,\PPTq}(\rho) && \textrm{(Proposition~\ref{prop:temp lowerbound})} \\
    &\leq E_{C,\PPTq}^\exact(\rho) && \textrm{(Equation~\eqref{eq:measure EC inequality})} \\
    &= E_{N}(\rho) = E_{D,\PPTq}^\exact(\rho) \quad && \textrm{(Theorem~\ref{thm:ent measure})} \\
    &\leq E_{D,\PPTq}(\rho) && \textrm{(Equation~\eqref{eq:measure ED inequality})} \\
    &\leq E_{N}(\rho). && \textrm{(Proposition~\ref{prop:ent distill <= en})}
\end{alignat}
\end{proof}
\end{document}

%% file: pretex.tex
\usepackage[dvipsnames]{xcolor}
\usepackage{framed}
\definecolor{shadecolor}{rgb}{0.9,0.9,0.9}

\usepackage{mathtools}
\usepackage{amsmath}
\usepackage[shortlabels]{enumitem}

\usepackage{graphicx,epic,eepic,epsfig,amsmath,latexsym,amssymb,verbatim,color}
 
\usepackage{amsfonts}       
\usepackage{nicefrac}       

\usepackage{amsmath}
\usepackage{bbm}

\usepackage{float}
\usepackage{tikz}
\usetikzlibrary{chains}
\usetikzlibrary{fit}
\usepackage{pgflibraryarrows}		
\usepackage{pgflibrarysnakes}		

\usepackage{epsfig}
\usetikzlibrary{shapes.symbols,patterns} 
\usepackage{pgfplots}

\usepackage[strict]{changepage}
\usepackage{hyperref}
\hypersetup{colorlinks=true,citecolor=blue,linkcolor=blue,filecolor=blue,urlcolor=blue,breaklinks=true}

\usepackage[marginal]{footmisc}
\usepackage{url}
\usepackage{theorem}

\newtheorem{definition}{Definition}
\newtheorem{proposition}{Proposition}
\newtheorem{lemma}[proposition]{Lemma}

\newtheorem{theorem}[proposition]{Theorem}

\def\squareforqed{\hbox{\rlap{$\sqcap$}$\sqcup$}}
\def\qed{\ifmmode\squareforqed\else{\unskip\nobreak\hfil
\penalty50\hskip1em\null\nobreak\hfil\squareforqed
\parfillskip=0pt\finalhyphendemerits=0\endgraf}\fi}
\def\endenv{\ifmmode\;\else{\unskip\nobreak\hfil
\penalty50\hskip1em\null\nobreak\hfil\;
\parfillskip=0pt\finalhyphendemerits=0\endgraf}\fi}
\newenvironment{proof}{\noindent \textbf{{Proof~} }}{\hfill $\blacksquare$}

\newcounter{remark}
\newenvironment{remark}[1][]{\refstepcounter{remark}\par\medskip\noindent%
\textbf{Remark~\theremark #1} }{\medskip}

\newcounter{example}

\mathchardef\ordinarycolon\mathcode`\:
\mathcode`\:=\string"8000
\def\vcentcolon{\mathrel{\mathop\ordinarycolon}}
\begingroup \catcode`\:=\active
  \lowercase{\endgroup
  \let :\vcentcolon
  }

\usepackage{cleveref}
\usepackage{graphicx}
\usepackage{xcolor}

\RequirePackage[framemethod=default]{mdframed}
\newmdenv[skipabove=7pt,
skipbelow=7pt,
backgroundcolor=darkblue!15,
innerleftmargin=5pt,
innerrightmargin=5pt,
innertopmargin=5pt,
leftmargin=0cm,
rightmargin=0cm,
innerbottommargin=5pt,
linewidth=1pt]{tBox}

\newmdenv[skipabove=7pt,
skipbelow=7pt,
backgroundcolor=red!15,
innerleftmargin=5pt,
innerrightmargin=5pt,
innertopmargin=5pt,
leftmargin=0cm,
rightmargin=0cm,
innerbottommargin=5pt,
linewidth=1pt]{rBox}

\newmdenv[skipabove=7pt,
skipbelow=7pt,
backgroundcolor=blue2!25,
innerleftmargin=5pt,
innerrightmargin=5pt,
innertopmargin=5pt,
leftmargin=0cm,
rightmargin=0cm,
innerbottommargin=5pt,
linewidth=1pt]{dBox}
\newmdenv[skipabove=7pt,
skipbelow=7pt,
backgroundcolor=darkkblue!15,
innerleftmargin=5pt,
innerrightmargin=5pt,
innertopmargin=5pt,
leftmargin=0cm,
rightmargin=0cm,
innerbottommargin=5pt,
linewidth=1pt]{sBox}
\definecolor{darkblue}{RGB}{0,76,156}
\definecolor{darkkblue}{RGB}{0,0,153}
\definecolor{blue2}{RGB}{102,178,255}
\definecolor{darkred}{RGB}{195,0,0}

\newcommand{\nc}{\newcommand}
\nc{\rnc}{\renewcommand}
\nc{\lbar}[1]{\overline{#1}}
\nc{\bra}[1]{\langle#1|}
\nc{\ket}[1]{|#1\rangle}
\nc{\ketbra}[2]{|#1\rangle\!\langle#2|}
\nc{\braket}[2]{\langle#1|#2\rangle}

\nc{\proj}[1]{| #1\rangle\!\langle #1 |}
\nc{\avg}[1]{\langle#1\rangle}
\nc{\smfrac}[2]{\mbox{$\frac{#1}{#2}$}}
\nc{\tr}{\operatorname{Tr}}
\nc{\ox}{\otimes}
\nc{\dg}{\dagger}
\nc{\dn}{\downarrow}
\nc{\cA}{{\cal A}}
\nc{\cB}{{\cal B}}
\nc{\cC}{{\cal C}}
\nc{\cD}{{\cal D}}
\nc{\cE}{{\cal E}}
\nc{\cF}{{\cal F}}
\nc{\cG}{{\cal G}}
\nc{\cH}{{\cal H}}
\nc{\cI}{{\cal I}}
\nc{\cJ}{{\cal J}}
\nc{\cK}{{\cal K}}
\nc{\cL}{{\cal L}}
\nc{\cM}{{\cal M}}
\nc{\cN}{{\cal N}}
\nc{\cO}{{\cal O}}
\nc{\cP}{{\cal P}}
\nc{\cQ}{{\cal Q}}
\nc{\cR}{{\cal R}}
\nc{\cS}{{\cal S}}
\nc{\cT}{{\cal T}}
\nc{\cU}{{\cal U}}
\nc{\cV}{{\cal V}}
\nc{\cX}{{\cal X}}
\nc{\cY}{{\cal Y}}
\nc{\cZ}{{\cal Z}}
\nc{\cW}{{\cal W}}
\nc{\csupp}{{\operatorname{csupp}}}
\nc{\qsupp}{{\operatorname{qsupp}}}
\nc{\var}{{\operatorname{var}}}
\nc{\rar}{\rightarrow}
\nc{\lrar}{\longrightarrow}
\nc{\polylog}{{\operatorname{polylog}}}
\nc{\wt}{{\operatorname{wt}}}
\nc{\av}[1]{{\left\langle {#1} \right\rangle}}
\nc{\supp}{{\operatorname{supp}}}

\nc{\argmin}{{\operatorname{argmin}}}
\DeclarePairedDelimiter{\norm}{\lVert}{\rVert}

\DeclarePairedDelimiter{\floor}{\lfloor}{\rfloor}

\def\x{\xi}

\def\o{\omega}

\def\O{\Omega}

\nc{\RR}{{{\mathbb R}}}
\nc{\CC}{{{\mathbb C}}}
\nc{\FF}{{{\mathbb F}}}
\nc{\NN}{{{\mathbb N}}}
\nc{\ZZ}{{{\mathbb Z}}}
\nc{\PP}{{{\mathbb P}}}
\nc{\QQ}{{{\mathbb Q}}}
\nc{\UU}{{{\mathbb U}}}
\nc{\EE}{{{\mathbb E}}}
\nc{\id}{{\operatorname{id}}}

\nc{\CHSH}{{\operatorname{CHSH}}}

\nc{\be}{\begin{equation}}
\nc{\ee}{{\end{equation}}}
\nc{\bea}{\begin{eqnarray}}
\nc{\eea}{\end{eqnarray}}
\nc{\<}{\langle}
\rnc{\>}{\rangle}
\nc{\rU}{\mbox{U}}

\nc{\ob}[1]{#1}

\nc{\SEP}{{\text{\rm SEP}}}
\nc{\NS}{{\text{\rm NS}}}
\nc{\LOCC}{{\text{\rm LOCC}}}
\nc{\PPT}{{\text{\rm PPT}}}
\nc{\EXT}{{\text{\rm EXT}}}
\nc{\OLOCC}{{\text{\rm 1-LOCC}}}
\nc{\Sym}{{\operatorname{Sym}}}

\nc{\ERLO}{{E_{\text{r,LO}}}}
\nc{\ERLOCC}{{E_{\text{r,LOCC}}}}
\nc{\ERPPT}{{E_{\text{r,PPT}}}}
\nc{\ERLOCCinfty}{{E^{\infty}_{\text{r,LOCC}}}}
\nc{\Aram}{{\operatorname{\sf A}}}

\newcommand{\eps}{\varepsilon}

\usepackage{tikz}
\usepackage{hyperref}
\hypersetup{colorlinks=true,citecolor=blue,linkcolor=blue,filecolor=blue,urlcolor=blue,breaklinks=true}

\makeatletter
\def\grd@save@target#1{%
  \def\grd@target{#1}}
\def\grd@save@start#1{%
  \def\grd@start{#1}}
\tikzset{
  grid with coordinates/.style={
    to path={%
      \pgfextra{%
        \edef\grd@@target{(\tikztotarget)}%
        \tikz@scan@one@point\grd@save@target\grd@@target\relax
        \edef\grd@@start{(\tikztostart)}%
        \tikz@scan@one@point\grd@save@start\grd@@start\relax
        \draw[minor help lines,magenta] (\tikztostart) grid (\tikztotarget);
        \draw[major help lines] (\tikztostart) grid (\tikztotarget);
        \grd@start
        \pgfmathsetmacro{\grd@xa}{\the\pgf@x/1cm}
        \pgfmathsetmacro{\grd@ya}{\the\pgf@y/1cm}
        \grd@target
        \pgfmathsetmacro{\grd@xb}{\the\pgf@x/1cm}
        \pgfmathsetmacro{\grd@yb}{\the\pgf@y/1cm}
        \pgfmathsetmacro{\grd@xc}{\grd@xa + \pgfkeysvalueof{/tikz/grid with coordinates/major step}}
        \pgfmathsetmacro{\grd@yc}{\grd@ya + \pgfkeysvalueof{/tikz/grid with coordinates/major step}}
        \foreach \x in {\grd@xa,\grd@xc,...,\grd@xb}
        \node[anchor=north] at (\x,\grd@ya) {\pgfmathprintnumber{\x}};
        \foreach \y in {\grd@ya,\grd@yc,...,\grd@yb}
        \node[anchor=east] at (\grd@xa,\y) {\pgfmathprintnumber{\y}};
      }
    }
  },
  minor help lines/.style={
    help lines,
    step=\pgfkeysvalueof{/tikz/grid with coordinates/minor step}
  },
  major help lines/.style={
    help lines,
    line width=\pgfkeysvalueof{/tikz/grid with coordinates/major line width},
    step=\pgfkeysvalueof{/tikz/grid with coordinates/major step}
  },
  grid with coordinates/.cd,
  minor step/.initial=.2,
  major step/.initial=1,
  major line width/.initial=2pt,
}
\makeatother

\usepackage{thmtools}
\usepackage{thm-restate}
\usepackage{etoolbox}
\makeatletter
\def\problem@s{}
\newcounter{problems@cnt}

\newcommand{\allproblems}{\problem@s}
\makeatother